\documentclass[graybox]{svmult}

\usepackage{mathptmx}       
\usepackage{helvet}         
\usepackage{courier}        
\usepackage{type1cm}        
\usepackage{makeidx}         
\usepackage{graphicx}        
\usepackage{multicol}        
\usepackage[bottom]{footmisc}

\usepackage{amsmath}
\usepackage{amssymb}
\usepackage{amscd}
\usepackage{indentfirst}

\def\1{{\bf 1}}

\def\R{{\mathbb R}}

\def\e{{\mathrm e}}

\makeindex             

\begin{document}

\title*{Scattering on leaky wires in dimension three}
\titlerunning{Scattering on leaky wires in dimension three}
\author{Pavel Exner and  Sylwia Kondej}
\authorrunning{Scattering on leaky wires in dimension three}
\institute{P.~Exner  \at {Department of Theoretical Physics, Nuclear Physics Institute, Czech Academy of Sciences, 25068 \v Re\v z near Prague, Czechia, and Doppler Institute for Mathematical Physics and Applied Mathematics, Czech Technical University, B\v rehov\'a 7, 11519 Prague, Czechia }\email{ exner@ujf.cas.cz} \\  S.~Kondej  \at {Institute of Physics, University of Zielona G\'ora, ul.\ Szafrana 4a,  65246 Zielona G\'ora, Poland }\email{ s.kondej@if.uz.zgora.pl} }

\maketitle

\abstract{We consider the scattering problem for a class of strongly singular Schr\"odinger operators in $L^2(\R^3)$ which can be formally written as $H_{\alpha,\Gamma}= -\Delta + \delta_\alpha(x-\Gamma)$ where $\alpha\in\R$ is the coupling parameter and $\Gamma$ is an infinite curve which is a local smooth deformation of a straight line $\Sigma\subset\R^3$. Using Kato-Birman method we prove that the wave operators $\Omega_\pm(H_{\alpha,\Gamma}, H_{\alpha,\Sigma})$ exist and are complete.
}


\section{Introduction}\label{s:intro}
\setcounter{equation}{0}

It is often said that when a great scientist says that something can be done, it can be done, while if the claim is that it cannot be done, he or she is usually wrong; sooner or later a younger one will come and do the impossible work earning a deserved fame. A nice illustration of this effect can be found in the biography of \emph{Tosio Kato} to the centenary of whom the present volume is dedicated. There are various testimonies \cite{Si18a} that John von Neumann who otherwise did so much for the mathematical foundations of quantum mechanics discouraged people from attempts to prove the self-adjointness of atomic Hamiltonians because he considered the task hopelessly beyond reach. Kato's elegant and in the matter of fact simple proof \cite{Ka51} was a starting point of the rigorous theory of Schr\"odinger operators which in the subsequent decades brought a plethora of results and managed to address fundamental questions such as those concerning the stability of matter \cite{LS}.

While this may be arguable the most important result of Tosio Kato, from the other point of view it is just one item in the broad spectrum of his achievements which extended also to some less well know directions \cite{Ok}. We believe that other contributions to this volume, in combination with recent reviews \cite{Si18a, Si18b} will provide a full picture showing how much mathematical physics owes to him. Different people may have different preferences but his \emph{opus magnum}, the monograph \cite{Ka}, will probably come to everybody's mind first. As anybody in the field, the present authors use it regularly and employed also other Kato's results, for instance his contribution to the theory of product formul{\ae} \cite{Ka78} that inspired us in the discussion of quantum Zeno dynamics \cite{EI05, EIK04, EINZ07}.

The result we are going to present in this note is based on a method at the origin of which Kato left his footprint and which bears his name together with that of Mikhail Birman. The starting point here were two of his 1957 papers \cite{Ka57a, Ka57b} which, together with the paper by Rosenblum \cite{Ro57} were the origin of the trace-class perturbation theory. Later substantial contributions were made by the others, the said Birman, Kuroda, Putnam, and Pearson, to name just the main ones --- for a description of the history we refer to the review \cite{Si18b} or the note to Sect.~XI.3 in \cite{RS} --- but the starting point was here.

In this note we are interested in Schr\"odinger operators with singular `potentials' supported by zero measure sets. In recent years they were studied as models of `leaky' quantum wires and networks made of them, cf.~\cite[Chap.~10]{EK} for an introduction to the subject and a bibliography. From the mathematical point of view the parameter which matters is the codimension of the interaction support. If the latter is one the operators can be treated naturally using the associated quadratic forms in the spirit of \cite{BEKS94}, for codimension two the problem is more subtle. The scattering problem in the codimension one case was discussed in \cite{EK05} where we considered the situation where the singular interaction support is a curve $\Gamma$ in the plane, or more generally a family of curves, which can regarded as a local deformation of a single straight line $\Sigma$. Under suitable regularity assumptions we proved there the existence and completeness of the wave operators.

In the present note we address a similar question in the codimension two case, for simplicity we restrict ourselves to the situation when $\Gamma$ is a single curve in $\R^3$ being a smooth local deformation of a straight line. Note that the scattering problem with singular interactions supported by curves in $\R^3$ has been considered recently\footnote{See also recent related results in \cite{CFP18, MP17}.} in \cite{BFKLR17}, however, our task here is different. The curves in the said paper were supposed to be finite and the Hamiltonian was compared to the one describing the free motion in the three-dimensional space, $-\Delta$ with the usual mathematical-physics license concerning the units. In our case the comparison operator can be formally written as $H_{\alpha,\Sigma}= -\Delta + \delta_\alpha(x-\Sigma)$. Using separation of variables and the well known result about two-dimensional point interactions \cite{AGHH}, we find easily that the spectrum is $[\xi_\alpha,\infty)$ where $\xi_\alpha<0$ is given by (\ref{eq-definitionxi}) below; in addition to motion at positive energies the system has now a guided mode in which the particle can move remaining localized in the vicinity of $\Sigma$. The scattering will now mean a comparison between $H_{\alpha,\Sigma}$ and the `full' Hamiltonian formally written as
\begin{equation}\label{eq-Hamiltonianformal}
  H_{\alpha, \Gamma}= \Delta + \delta_\alpha (x-\Gamma)\,;
\end{equation}
a rigorous definition of these operators will be given in the next section. Our aim is to show that the scattering is well defined in this setting, in other words, that the wave operators for the given pair exist and are complete. In fact, the wave operators are also asymptotically complete as one has $\sigma_\mathrm{sc} (H_{\alpha, \Gamma})=\emptyset$ under our assumptions, but we leave the proof of this property together with extensions of the result to a wider class of the interaction supports $\Gamma$ to a subsequent paper.

\section{Preliminaries: the operator}\label{s:prelim}
\setcounter{equation}{0}

First we have to introduce the main notions. Let $\Gamma \subset \R^3$ be an infinite curve of class $C^1$ and piecewise $C^2$ which coincides asymptotically with a straight line $\Sigma$ in a sense to be made precise below. With the usual abuse of notation we regard $\Gamma$ both as a map $\R\to\R^3$ and its image. Furthermore, without loss of generality we may fix $\Sigma := \{(x_1,0,0)\,:\,x_1 \in \R \}$ and to parametrize $\Gamma$ by its arc length $s$ so that we have $|\dot\Gamma|=1$ and $|\Gamma(s)-\Gamma(s')|\le|s-s'|$. We will also suppose that
 \begin{equation} \label{no-u}
\exists c\in (0,1) \quad \text{such that} \quad |\Gamma(s)-\Gamma(s')|\ge c|s-s'|\quad \text{for} \quad \forall s,s'\in\R\,,
 \end{equation}
which means, in particular, that the curve $\Gamma$ has no self-intersections and that it cannot be of a U-shape form.

Our next task is to introduce the Hamiltonian, which will be a singular Schr\"odinger operator with an interaction supported by the curve $\Gamma$, in other words a singular perturbation of the `free' operator $H_0$ which is the Laplacian in $L^2(\R^3)$ with the natural domain. There are  various ways to do that using, for instance, quadratic forms or a Krein-type formula \cite{BFKLR17, BT92, Pos01, Te90}. For the purpose of the present paper we recall the method employed in \cite{EK02} inspired by the classical definition of the two-dimensional $\delta$ interaction \cite{AGHH}; its advantage is that it has a local character. The curve regularity allows us to associate with $\Gamma$, apart from a discrete set, the Frenet's frame, i.e. the triple $(t(s),b(s),n(s))$ of the tangent, binormal and normal vectors, which are by assumption piecewise continuous functions of $s$. Moreover, at the discontinuity points of $\ddot\Gamma$ the Frenet frame limits from the two sides differ by a rotation around the tangent vector, hence one can construct a globally smooth coordinate system and, with an abuse of notation, employ the symbols $b(s),n(s)$ for the rotated binormal and normal, respectively.

Using this system, we may further introduce a family of `shifted' curve: given $\xi ,\eta \in \R$ we denote $r=(\xi ^{2}\!+\!\eta^{2})^{1/2}$ and set
 $$ 
 \Gamma^{\xi \eta } _{r}:= \{\,\Gamma ^{\xi \eta }_{r}(s):=\Gamma (s)+\xi b(s)+\eta n(s)\,,\:s\in\R\}\,,
 $$ 
in view of \eqref{no-u} and the smoothness of $\Gamma$ there is an $r_{0}>0$ such that $\Gamma^{\xi \eta }_{r}\cap \Gamma =\emptyset $ holds for all $r< r_{0}$. This allows us to define generalized boundary values of a function $f\in H_\mathrm{loc}^{2}(\R^{3}\setminus \Gamma )$ using its restriction ${f\!\upharpoonright}_{\Gamma^{\xi\eta}_{r}}(s)$ to $\Gamma ^{\xi\eta}_{r}$ which is by assumption well defined for $0<r<r_0$. We shall say that a function $f\in H_\mathrm{loc}^{2}(\R^{3}\setminus \Gamma )\cap L^{2}(\R^{3})$\ belongs to $\Upsilon$ if the limits
 \begin{align*}
 \Xi (f)(s) &:= -\lim_{r \to 0}\frac{1}{\ln r }
 {f\!\upharpoonright}_{\Gamma^{\xi \eta }_{r}}(s), \\
 \Omega (f)(s) &:= \lim_{r \to 0}
 \left[{f\!\upharpoonright}_{\Gamma^{\xi \eta }_{r}} (s)+\Xi (f)(s)\ln r\right]
 \end{align*}
exist a.e. in $\R$ independently of the direction $\frac{1}{r}(\xi,\eta)$ in which they are taken and belong to $ L^{2}(\R)$ as functions of $s$. This makes it possible to define the sought singular Schr\"odinger operator as the restriction of the Laplacian acting on $\R^3\setminus\Gamma$ to a suitable subset of $\Upsilon$.

To be specific, we fix a nonzero $\alpha\in\R$ and define the operator $H_{\alpha,\Gamma}$ as follows,
 \begin{subequations}\label{bcond3}
 \begin{align}
 D(H_{\alpha,\Gamma}) &=
 \Upsilon _{\alpha }:= \{\,g\in \Upsilon :\:2\pi \alpha \Xi
 (g)(s)=\Omega (g)(s) \,\}, \\
 H_{\alpha,\Gamma}f &= -\Delta f \quad\mathrm{for}
 \quad x \in \R^{3}\setminus \Gamma.
 \end{align}
 \end{subequations}
It was shown in \cite{EK02} that such an operator is self-adjoint. Note that the absence of a singular interaction means that the singular boundary value $\Xi(f)$ vanishes identically, in other words, the free operator $H_0$ corresponds to $\alpha=\infty$. This fact leads some authors to write the operator in question as $-\Delta -\frac{1}{\alpha}\delta(\cdot-\Gamma)$, see e.g. \cite{BFKLR17}. This, however, does not fit well with the fact that the two-dimensional $\delta$ interaction is `always attractive', hence we avoid such formal expressions showing the interaction `strength' and restrict ourselves to the definition \eqref{bcond3} in the spirit of \cite[Sec.~I.5]{AGHH}.

Before proceeding further, let us say a few words about the spectrum of $H_{\alpha,\Gamma}$. If the interaction support is a straight line, $\Gamma= \Sigma$, one finds it easily by separation of variables: it is absolutely continuous and equal to
 $$ 
 \sigma (H_{\alpha,\Gamma})= [\xi_\alpha,\infty),
 $$ 
where
\begin{equation}\label{eq-definitionxi}
  \xi_\alpha = -4\,\e^{2(-2\pi \alpha +\psi(1))}
\end{equation}
is the eigenvalue of the corresponding one-center two-dimensional $\delta$ interaction, with $-\psi(1)\approx 0.57721$ being the Euler-Mascheroni constant. For a non-straight $\Gamma$ the spectrum of $H_{\alpha,\Gamma}$ may be different and depends in general on the geometry of $\Gamma$. One of the interesting situations concerns curves that are asymptotically straight. In \cite{EK02}, for instance, we assumed that there are $\omega \in (0,1)$, $\mu \geq 0$ and $\varepsilon,d>0 $ such that for all $(s,s^{\prime })\in S_{\omega ,\varepsilon }$ we have
 \begin{equation} \label{a2}
 1-\frac{|\Gamma(s)-\Gamma(s^{\prime })|}{|s-s^{\prime}|}\leq
 \frac{d|s-s^{\prime}|}{(|s-s^{\prime}|\!+\!1)
 (1+(s^{2}\!+\!s^{\prime 2})^{\mu })^{1/2}},
 \end{equation}
where $S_{\omega ,\varepsilon }$ is the subset of $\R^2$ consisting of points $(s,s')$ satisfying $\omega < \frac{s}{s^{\prime}} <\omega^{-1}$ if $|s+s^{\prime}| > \varepsilon\, \frac{1+\omega}{1-\omega}$ and $|s-s^{\prime}| <\varepsilon$ if $|s+s^{\prime}| < \varepsilon\, \frac{1+\omega}{1-\omega}$. If this assumption is satisfied with some $\mu>\frac12$, together with \eqref{no-u}, then the essential spectrum is preserved,
$$
\sigma_\mathrm{ess}(H_{\alpha,\Gamma}) =[\xi_\alpha,\infty ),
$$
and in addition, the operator $H_{\alpha,\Gamma}$ has a non-void discrete spectrum whenever the deformation is nontrivial, $\Gamma\ne\Sigma$.

\section{Preliminaries: the resolvent}\label{s:prelim}
\setcounter{equation}{0}

In what follows we adopt a more restrictive assumption about the curve, namely we suppose that there exists a compact set
$M\subset\R^3$ such that
\begin{equation} \label{compact}
\Gamma \setminus \Sigma  \subset M\,.
\end{equation}
To analyze the scattering problem for the pair $(H_{\alpha , \Gamma },H_{\alpha,\Sigma})$ we need to know more about the resolvent of singular Schr\"odinger operator \eqref{bcond3}. In analogy with the considerations of \cite{EK02} we begin from the embedding of the free resolvent $R^z:= (-\Delta -z )^{-1}\,:\, L^2 (\R^3) \to W^{2,2} (\R^3 )$ to $L^2 (\Gamma )$. It is sufficient to restrict the spectral parameter $z$ to negative real values, hence we consider $z=-\kappa^2$ with $\kappa >0$ and denote $\mathbf{R}^{\kappa} = R^{-\kappa^2}$. It is well known that $\mathbf{R}^\kappa $ is integral operator with the kernel determined by the function
$$
G^\kappa  (x):= \frac{\e^{-\kappa |x|}}{4\pi |x|}\,.
$$
Specifically, $\mathbf{\breve {R} }^\kappa _\Gamma\,:\, L^{2} (\Gamma )\to L^2 (\R^3  )$ acts as $\mathbf{\breve{R}}^\kappa _\Gamma f:= \int_{\R^3}G^\kappa  (\cdot -x) f (x) \delta (x-\Gamma ) \mathrm{d}x$, and furthermore, we define $\mathbf{R}_\Gamma ^\kappa \,:\,  L^2 (\R^3) \to L^2 (\Gamma)$ as the adjoint of $\mathbf{\breve {R} }^\kappa _\Gamma$. To find out the resolvent of $H_{\alpha, \Gamma }$, we define the operator
$\check{\mathbf{T}}^\kappa \,:\, L^2 (\Gamma) \to L^{2} (\Gamma )$ by
\begin{equation}\label{eq-checkT}
(\check{\mathbf{T}}^\kappa  f)(s) = - \frac{1}{(2\pi)^2}\int_{\R} \ln (p^2 +\kappa^2)^{1/2}\, \e^{ips} \hat{f}(p)\, \mathrm{d}p,
\end{equation}
where $\hat{f}$ stands for the Fourier transform of $f$ and the maximal domain of this operator is $D(\check{\mathbf{T}}^\kappa  )=\{f\,:\, \check{\mathbf{T}}^\kappa f\in L^2 (\Gamma )\}$. Furthermore, we set
\begin{equation}\label{eq-T}
\mathbf{T}^\kappa = \check{\mathbf{T}}^\kappa + \frac{1}{2\pi} (\ln 2 +\psi(1)),
\end{equation}
where $\psi(1)$ is, up to the sign, the Euler-Mascheroni constant mentioned above. Finally, we define the integral operator $\mathbf{B}^\kappa \,:\, L^2 (\Gamma ) \to L^2 (\Gamma )$ with the kernel of the form
$$
\mathbf{B}^\kappa  (s,s'):= G^\kappa (\Gamma (s)-\Gamma (s'))- G^\kappa (s-s'),
$$
and the operator
\begin{equation} \label{Krein}
\mathbf{Q}^\kappa := \mathbf{T}^\kappa +\mathbf{B}^\kappa \,:\, D:= D(\check{\mathbf{T}}^\kappa ) \to L^2(\Gamma )\,.
\end{equation}
Note that the operator $\mathbf{B}^\kappa$ is positive because the function $G^\kappa$ is monotonous and by assumption we have $|\Gamma (s)-\Gamma (s')|\le|s-s'|$; this fact was crucial in \cite{EK02} to prove that a curve deformation gives rise to the existence of a discrete spectrum of $H_{\alpha,\Gamma}$. By \cite[Thm.~2.1]{EK02} the operator $\alpha -\mathbf{Q}^\kappa \,:\, L^2 (\Gamma )\to L^2 (\Gamma )$ is invertible for all $\kappa$ large enough and
\begin{equation} \label{resolvent}
\mathbf{R}^\kappa_{\alpha, \Gamma} =  \mathbf{R}^\kappa + \mathbf{\breve{R} }_\Gamma^\kappa(\alpha -\mathbf{Q}^\kappa )^{-1}  \mathbf{R }_\Gamma^\kappa
\end{equation}
is the resolvent of $H_{\alpha, \Gamma }$. It is not by a chance that this expression has a Krein-like form because $H_{\alpha, \Gamma }$ is a self-adjoint extension of the symmetric operator $-\Delta \,: \, C^\infty_0 (\R^3 \setminus \Gamma )\to L^2(\R^3)$. Note also that the geometric perturbation is encoded in the part $\mathbf{B}^\kappa$ of \eqref{Krein}: we have $\mathbf{B}^\kappa=0$ if $\Gamma =\Sigma $, and consequently, $\mathbf{Q}^\kappa =\mathbf{T}^\kappa$ holds in this case.  The resolvent expression \eqref{resolvent} is a tool to prove the spectral properties of $H_{\alpha, \Gamma }$ mentioned at the end of the preceding section.

\begin{lemma}
The operator $(\alpha -\mathbf{Q}^\kappa )^{-1}$ is bounded for all $\kappa $ large enough.
\end{lemma}
\begin{proof}
It follows from (\ref{eq-checkT}) that
$$
\|\mathbf{\check{T}}^\kappa f \|^2 = \frac{1}{4 (2\pi )^3 } \int_{\R} \big(\ln (p^2 +\kappa^2)\big)^2\, |\hat{f}(p)|^2\, \mathrm{d}p\,
$$
and therefore for all $\kappa $ large enough we have
$$
\|\mathbf{T}^\kappa f \|^2 \geq C (\ln \kappa )^2 \|f\|^2
$$
with a suitable constant $C$. For the sake of simplicity we use the symbol $C$ as a generic positive constant which may vary case from case. Furthermore, by \cite[Lemma~5.3]{EK02} the operator $\mathbf{B}^\kappa $ belongs to the Hilbert--Schmidt class under assumption \eqref{a2}, and therefore, \emph{a fortiori}, if we assume \eqref{compact}. This allows to conclude that
$$
\left\| (\alpha -\mathbf{Q}^\kappa )f \right\|^2 =
\left( (\alpha -\mathbf{T}^\kappa -\mathbf{B}^\kappa)f,(\alpha  -\mathbf{T}^\kappa -\mathbf{B}^\kappa )f \right) \,
\geq
C (\ln \kappa )^2 \|f\|^2
$$
with another constant $C$. On the other hand, we know that operator  $(\alpha -\mathbf{Q}^\kappa )^{-1}$ exists and from the above
inequality we can conclude that it is bounded. \hfill$\Box$
\end{proof}

\section{Existence and Completeness of Wave Operators}\label{s:wave}
\setcounter{equation}{0}

Now we are able to pass to our main task in this note, namely the existence and completeness of the wave operators given by
$$
\Omega_\pm(H_{\alpha,\Gamma}, H_{\alpha,\Sigma}) := \textrm{s\,-}\!\!\!\lim_{t\to\pm\infty} \e^{iH_{\alpha,\Gamma}t}\, \e^{-iH_{\alpha,\Sigma}t}\,,
$$
where we can skip the projection $E_\mathrm{ac}(H_{\alpha,\Sigma})$ usually appearing in the definition because the spectrum of $H_{\alpha,\Sigma}$ is purely absolutely continuous as we have recalled above. For notational convenience, we decompose the line $\Sigma $ into three parts,
$$
\Sigma = \Sigma_M \cup \Sigma_+ \cup \Sigma_-\,,
$$
where $\Sigma_M := M\cap \Sigma$ and $\Sigma_\pm $ are the straight components of $\Sigma \setminus M$ which correspond, respectively, to $x_1 \to \pm \infty $ in the chosen coordinate system. Without loss of generality we may assume that $(0,0,0)\in M$ and 
$$
\Sigma_\pm =\{ x\,:\,x= (x_1, 0, 0),\;  x_1\in (-\infty , x_-) \cup
(x_+, \infty )\;\;\text{and}\;\; x_\pm \gtrless 0 \}\,.
$$
In a similar way one can dissect the curve $\Gamma$ into three parts, $\Gamma_M$ and $\Gamma_\pm =\Sigma_\pm $. The inverses of the `full' and `free' Birman--Schwinger operators,  $(\alpha -\mathbf{Q} ^\kappa )^{-1}$ and $(\alpha -\mathbf{T}^\kappa )^{-1}$, act respectively in $L^2 (\Gamma )$ and $L^2 (\Sigma )$. To compare the resolvents of $H_{\alpha, \Gamma }$ and $H_{\alpha, \Sigma}$ we introduce the following embeddings,
$$
(\alpha -\mathbf{Q}^\kappa ) ^{-1}_{\Gamma_i \Gamma_j}:= \chi _{ \Gamma_i}(\alpha -\mathbf{Q}^\kappa )^{-1} \chi _{ \Gamma_j}\,:\, L^2 (\Gamma_j) \to L^2 (\Gamma_i)\,,
$$
where $i,j=\pm,M$, $\:\chi _{ \Gamma_i}$ is the characteristic functions of $\Gamma_i$, and in the analogous way we define $(\alpha -\mathbf{T}^\kappa ) ^{-1}_{\Sigma_i \Sigma_j}$.
Let us now consider the resolvent difference
$$
\mathbf{R}^\kappa_{\alpha, \Gamma} -
\mathbf{R}^\kappa_{\alpha, \Sigma }=  \mathbf{\breve{R} }_\Gamma^\kappa(\alpha -\mathbf{Q}^\kappa )^{-1}  \mathbf{R }_\Gamma^\kappa - \mathbf{\breve{R} }_\Sigma^\kappa(\alpha -\mathbf{T}^\kappa )^{-1}  \mathbf{R }_\Sigma^\kappa\,.
$$
Using the obvious fact
$$
(\alpha -\mathbf{Q}^\kappa ) ^{-1}_{\Gamma_\pm  \Gamma_\pm } = (\alpha -\mathbf{T}^\kappa ) ^{-1}_{\Sigma _\pm  \Gamma_\pm }\,,
$$
we are coming to the conclusion that
\begin{equation}\label{eq-resol}
\mathbf{R}^\kappa_{\alpha, \Gamma} -
\mathbf{R}^\kappa_{\alpha, \Sigma }=  \sum_{i,j\in \mathcal{X}} \mathbf{\breve{R} }_{\Gamma_i}^\kappa(\alpha -\mathbf{Q}^\kappa )^{-1}_{\Gamma_i \Gamma_j }  \mathbf{R }_{\Gamma_j}^\kappa -\\
\sum_{i,j\in \mathcal{X}} \mathbf{\breve{R} }_{\Sigma _i}^\kappa(\alpha -\mathbf{T}^\kappa )^{-1}_{\Sigma_i \Sigma_j }  \mathbf{R }_{\Sigma_j}^\kappa\,,
\end{equation}
where $ \mathcal{X}:= \{(i,j)\,:\, i,j=+,-,M \,\wedge \,(i,j)\neq  (+,+), (-,-)\}$. This will be used to prove the following result:
\begin{theorem} \label{thm:main}
The operator $\mathbf{R}^\kappa_{\alpha, \Gamma} -\mathbf{R}^\kappa_{\alpha, \Sigma }$ belongs to the trace class for all $\kappa$ large enough.
\end{theorem}

Let us start from an auxiliary claim:
\begin{lemma} \label{le-G2}
We have
\begin{equation}\label{eq-appen}
  \int_{\R^3} G^\kappa (y-x) G^\kappa (x-z)\,\mathrm{d}x = \frac{1}{8\pi \kappa }\, \e^{-\kappa |y-z|}\,.
\end{equation}
\end{lemma}
\begin{proof}
We use the Fourier representation of the Green function
$$
\int_{\R^3} G^\kappa (y-x)\, f(x)\,\mathrm{d}x =\frac{1}{(2\pi )^{3/2}} \int_{\R^3}\frac{\e^{ipy}}{p^2 +\kappa ^2}\, \hat{f}(p)\, \mathrm{d}p
$$
and apply it to the Green function itself,
$$
\int_{\R^3} G^\kappa (y-x) G^\kappa (x-z)\,\mathrm{d}x =
\frac{1}{(2\pi )^3} \int_{\R^3}\frac{\e^{ip(y-z)}}{(p^2 +\kappa ^2)^2}  \,\mathrm{d}p\,.
$$
Performing the integration over angles in the integral on the right-hand side we get
$$
I=\int_{\R^3}\frac{\e^{ipy}}{(p^2 +\kappa ^2)^2}\,  \mathrm{d}p=
\frac{4\pi }{|y|} \int_0^\infty \frac{p\,\sin p|y|}{(p^2+\kappa^2)^2}\,\mathrm{d}p =
\frac{2\pi }{i|y|} \int_\R \frac{p\,\e^{ip|y|}}{(p^2+\kappa^2)^2}\,\mathrm{d}p\,.
$$
Next we extend in a standard way the integration forming the contour by adding upper semicircle and using the Jordan's lemma which implies that the integral over the semicircle vanishes in the limit of infinite radius. This gives
$$
I= \frac{2\pi }{i|y|} \oint \frac{z\,\e^{iz|y|}}{(z^2+\kappa^2)^2}\,\mathrm{d}z =
\frac{4\pi ^2 }{|y|}\sum_{upper \,\,halfplane } \mathrm{Res} \,\frac{z\e^{iz|y|}}{(z^2 +\kappa^2)^2}\,.
$$
Using now the generalized Cauchy integral formulae one gets
$$
I= \frac{4\pi ^2 }{|y|} \left(\frac{d}{dz}\frac{z\e^{iz|y|} }{(z+i\kappa)^2}\right)\Big|_{z=i\kappa}= \frac{\pi^2 }{\kappa }\,\e^{-\kappa |y|}\,.
$$
Putting these results together we arrive at the formula (\ref{eq-appen}). \hfill$\Box$
\end{proof}

\medskip

\emph{Proof of Theorem~\ref{thm:main}.}
Let us pick one of the components of $\mathbf{R}^\kappa_{\alpha, \Gamma} - \mathbf{R}^\kappa_{\alpha, \Sigma }$, for instance, $\mathbf{\breve{R} }_{\Gamma_+}^\kappa(\alpha -\mathbf{Q}^\kappa )^{-1}_{\Gamma_+ \Gamma_- }  \mathbf{R }_{\Gamma_-}^\kappa $. The symbol $\mathcal{B}_\delta $ will conventionally denote the ball of radius $\delta $ centered at the origin and $\chi_{\mathcal{B}_\delta}$ stands for the characteristic function of the ball $\mathcal{B}_\delta$. We define the `cut-off' operator family
$$
S_\delta^{+-} \equiv S_\delta  := \chi_{\mathcal{B}_\delta}\left( \mathbf{\breve{R} }_{\Gamma_+}^\kappa(\alpha -\mathbf{Q}^\kappa )^{-1}_{\Gamma_+ \Gamma_- }  \mathbf{R }_{\Gamma_-}^\kappa \right) \chi_{\mathcal{B}_\delta}
$$
and ask for its weak limit as $\delta \to \infty$. We have, in particular,
\begin{equation} \label{eq-trace}
\int_{\R^3 } S_\delta  (x,x)\mathrm{d}x = \int_{\R^3 } \left( G^\kappa (\cdot -x)\chi_{\mathcal{B}_\delta }(x) ,
(\alpha -\mathbf{Q}^\kappa )^{-1}_{\Gamma_+ \Gamma_- } G^\kappa (\cdot -x)\chi_{\mathcal{B}_\delta } (x)
\right)_{L^2 (\Gamma_+)} \,\mathrm{d}x\,.
\end{equation}
Using now the Lebesgue's dominated convergence theorem in combination with Lemma~\ref{le-G2} and the equivalence that $\e^{-\kappa |x-y|} = \e^{-\kappa (|x   |+|y|)}$ holds, in view of that fact that $\Gamma_\pm = \Sigma_\pm =\{ x\,:\,x= (x_1, 0, 0),\;  x_1\in (-\infty , x_-) \cup
(x_+, \infty )\;\;\text{and}\;\; x_\pm \gtrless 0 \}$, we obtain
\begin{eqnarray} \label{eq-psidelta}
\lim_{\delta \to \infty}\int_{\R^3 } S_\delta  (x,x)\,\mathrm{d}x =
\frac{\pi ^4}{\kappa^2 }\int_{\Gamma_+ } e^{-\kappa |s |} \left( (\alpha -\mathbf{Q}^\kappa )^{-1}_{\Gamma_+ \Gamma_- }  e^{-\kappa |\cdot   |} \right) (s)\,\mathrm{d}s \,.
\end{eqnarray}
Using further the boundedness of $(\alpha -\mathbf{Q}^\kappa )^{-1}_{\Gamma_+ \Gamma_- }$ in combination with Schwarz inequality we get from (\ref{eq-psidelta}) the following estimate
$$
\lim_{\delta\to\infty}\int_{\R^3 } S_\delta  (x,x) \mathrm{d}x\leq C \frac{\pi ^4}{\kappa^2 } \|e^{-\kappa |\cdot |}\|_{L^2 (\Gamma_+)} \|e^{-\kappa |\cdot |}\|_{L^2 (\Gamma_-)}\,,
$$
where the constant $C$ is the norm of $(\alpha -\mathbf{Q}^\kappa )^{-1}_{\Gamma_+ \Gamma_- }$. We want to conclude that
$$ \int_{\R^3 } S_\delta  (x,x)\mathrm{d}x \to \mathrm{Tr} \,\,\mathbf{\breve{R} }_{\Gamma_+}^\kappa(\alpha -\mathbf{Q}^\kappa )^{-1}_{\Gamma_+ \Gamma_- }  \mathbf{R }_{\Gamma_-}^\kappa
$$
as $\delta\to\infty$ and to show in this way that the operator $\mathbf{\breve{R} }_{\Gamma_+}^\kappa(\alpha -\mathbf{Q}^\kappa )^{-1}_{\Gamma_+ \Gamma_- }  \mathbf{R }_{\Gamma_-}^\kappa$ belongs to the trace class. According to the lemma following Theorem~XI.31 in \cite{RS} the trace can be expressed through the integral of the kernel diagonal provided the latter is continuous in both arguments and the operator is positive. The continuity was mentioned already, the positivity follows from the fact that the operator $\alpha -\mathbf{Q}^\kappa$ is positive from all $\kappa$ large enough, cf. \cite[Lemma~5.5]{EK02}.

Let us next consider the component of~(\ref{eq-resol}) referring to the operator acting between the spaces $L^2 (\Gamma_+)$ and $L^2 (\Gamma_M)$. We put
$$
S_\delta^{M+} \equiv S_\delta  := \chi_{\mathcal{B}_\delta}\left( \mathbf{\breve{R} }_{\Gamma_M}^\kappa(\alpha -\mathbf{Q}^\kappa )^{-1}_{\Gamma_M \Gamma_+ }  \mathbf{R }_{\Gamma_+ }^\kappa \right) \chi_{\mathcal{B}_\delta}\,.
$$
Applying again Lemma~\ref{le-G2} in combination with the Lebesgue's dominated convergence theorem one obtains
\begin{equation} \label{eq-psideltaa}
\lim_{\delta \to\infty}\int_{\R^3 } S_\delta  (x,x)\,\mathrm{d}x =
\frac{\pi ^4}{\kappa^2 }\int_{\Gamma_M }  \left( (\alpha -\mathbf{Q}^\kappa )^{-1}_{\Gamma_M \Gamma_+}  \e^{-\kappa |\Gamma (\cdot) - \Gamma (s)  |} \right) \mathrm{d}s \,.
\end{equation}
Using further the boundedness of  $(\alpha -\mathbf{Q}^\kappa )^{-1}_{\Gamma_M \Gamma_+}$ and the continuous imbedding of spaces $L^2 (\Gamma_M) \hookrightarrow L^1 (\Gamma_M) $ together with the Fubini's theorem and Schwarz inequality we infer that
\begin{eqnarray} \label{eq-psideltaaa}
  \nonumber
\Big| \lim_{\delta \to \infty  }\int_{\R^3 } S_\delta  (x,x)\,\mathrm{d}x \Big| \leq
\frac{\pi ^4}{\kappa^2 }\int_{\Gamma_M } \Big| \left( (\alpha -\mathbf{Q}^\kappa )^{-1}_{\Gamma_M \Gamma_+}  \e^{-\kappa |\Gamma (\cdot) - \Gamma (s)  |} \right) \Big| \mathrm{d}s
\\ \label{eq-aux2}
\leq \frac{\pi ^4}{\kappa^2 }|\Gamma_M |\left(  \int_{\Gamma_M}    \left| (\alpha -\mathbf{Q}^\kappa )^{-1}_{\Gamma_M \Gamma_+}
\e^{-\kappa |\Gamma (\cdot ) - \Gamma (s)  |} \right|^2 \mathrm{d}s \right)^{1/2}\,.
\end{eqnarray}
The integral on the right-hand side of~(\ref{eq-aux2}) is finite because the integrated function belongs to $L^2 (\Gamma_M)$. This implies that $\mathbf{\breve{R} }_{\Gamma_M}^\kappa(\alpha -\mathbf{Q}^\kappa )^{-1}_{\Gamma_M \Gamma_+ }  \mathbf{R }_{\Gamma_+ }^\kappa$ belongs to the trace class in the same way as above.

The remaining components of $\mathbf{R}^\kappa_{\alpha, \Gamma} - \mathbf{R}^\kappa_{\alpha, \Sigma }$ contributing to formula~(\ref{eq-resol}) can be dealt with in the analogous way. The only terms which do not allow for such a treatment are those containing $(\alpha -\mathbf{Q}^\kappa )^{-1}_{\Gamma_\pm  \Gamma_\pm }$, however, they cancel when we subtract $\mathbf{R}^\kappa_{\alpha, \Gamma}$ from $\mathbf{R}^\kappa_{\alpha, \Sigma }$. Concluding the above discussion we thus find that the difference $\mathbf{R}^\kappa_{\alpha, \Gamma} - \mathbf{R}^\kappa_{\alpha, \Sigma }$ is a trace class operator for all $\kappa$ large enough what we have set out to prove. \hfill$\Box$

\medskip

Now we are in position to present the result indicated in the introduction:

\begin{corollary}
In the stated assumptions, the wave operators $\Omega_\pm(H_{\alpha,\Gamma}, H_{\alpha,\Sigma})$ exist and are complete.
\end{corollary}
\begin{proof}
In view of Theorem~\ref{thm:main} the claim follows immediately from Kuroda-Birman theorem, cf.~\cite[Thm.~XI.9]{RS}. \hfill$\Box$
\end{proof}

\subsection*{Acknowledgements}
The research was supported by the Czech Science Foundation (GA\v{C}R) within the project 17-01706S.

\end{document}